\documentclass[journal,twoside,print]{ieeecolor}
\usepackage{lcsys}
\usepackage{booktabs}
\usepackage{cite}
\usepackage{amsmath,amssymb,amsfonts}
\usepackage{algorithmic}
\usepackage{graphicx}   
\usepackage{caption}    
\usepackage{subcaption} 
\usepackage{url}
\usepackage{amsmath}
\usepackage{amsthm, amssymb}
\theoremstyle{definition}
\newtheorem{definition}{Definition}
\theoremstyle{assumption}
\newtheorem{assumption}{Assumption}
\theoremstyle{problem}

\theoremstyle{lemma}
\newtheorem{lemma}{Lemma}
\theoremstyle{remark}
\newtheorem{remark}{Remark}
\theoremstyle{theorem}
\newtheorem{theorem}{Theorem}

\newcommand{\norm}[1]{\left \lVert #1 \right \rVert}

\usepackage{breqn}

\usepackage{textcomp}
\def\BibTeX{{\rm B\kern-.05em{\sc i\kern-.025em b}\kern-.08em
    T\kern-.1667em\lower.7ex\hbox{E}\kern-.125emX}}
\begin{document}
\title{Direct Data Driven Control Using Noisy Measurements}
\author{Ramin Esmzad, Gokul S. Sankar, Teawon~Han, Hamidreza Modares, \IEEEmembership{Senior, IEEE}
\thanks{This work is supported by the Ford Motor Company–Michigan State University Alliance. }
\thanks{R. Esmzad and H. Modares are with Michigan State University, East Lansing, MI 48824 USA (e-mails:  (esmzadra, modaresh)@msu.edu).}
\thanks{G. S. Sankar is with Ford Motor Company, Dearborn, MI 48124, USA. (email: ggowris2@ford.com)}
\thanks{T. Han T. Han was with Ford Motor Company (e-mail: teawon.han@gmail.com)}
}

\maketitle

\begin{abstract}
This paper presents a novel direct data-driven control framework for solving the linear quadratic regulator (LQR) under disturbances and noisy state measurements. The system dynamics are assumed unknown, and the LQR solution is learned using only a single trajectory of noisy input-output data while bypassing system identification. Our approach guarantees mean-square stability (MSS) and optimal performance by leveraging convex optimization techniques that incorporate noise statistics directly into the controller synthesis. First, we establish a theoretical result showing that the MSS of an uncertain data-driven system implies the MSS of the true closed-loop system. Building on this, we develop a robust stability condition using linear matrix inequalities (LMIs) that yields a stabilizing controller gain from noisy measurements. Finally, we formulate a data-driven LQR problem as a semidefinite program (SDP) that computes an optimal gain, minimizing the steady-state covariance. Extensive simulations on benchmark systems—including a rotary inverted pendulum and an active suspension system—demonstrate the superior robustness and accuracy of our method compared to existing data-driven LQR approaches. The proposed framework offers a practical and theoretically grounded solution for controller design in noise-corrupted environments where system identification is infeasible.
\end{abstract}

\begin{IEEEkeywords}
Data-driven control, Noisy measurements, Linear Quadratic Regulator (LQR), Mean-square stability (MSS), Semidefinite programming (SDP), Linear matrix inequalities (LMI), Input-output data, Robust control.

\end{IEEEkeywords}

\section{Introduction}
\label{sec:introduction}
\IEEEPARstart{D}{irect} data-driven control has recently gained a surge of interest due to its control-oriented approach to solving control design problems \cite{9867266,gravell2022robust,ronnmark2024implementation}. That is, controller parameters are learned directly using input-output or input-state trajectories, without explicitly constructing a predictive model of the system. 
Bypassing system identification allows for leveraging the collected data to achieve what is best for the control objectives rather than using the data to fit a predictive model. This is because the best predictive model does not necessarily lead to the best control performance. 
 
 While direct data-driven methods have shown promising results, much of the existing work assumes the availability of high-quality data with full-state measurements.  This includes learning using noise-free state measurements \cite{SASSELLA20231382,DEPERSIS2021109548,ESMZAD2025112197,10644622,9705109,8960476,DAI20231376,miller2024ddsrs,DAI2023111041,8933093}. However, in reality, partial state measurements are available and these measurements are also corrupted with noise. To address this gap, recent efforts have focused on extending direct data-driven control frameworks to handle noisy state measurements or partial measurements \cite{10549788,10124991}. 
More specifically, \cite{10549788} presents a data-driven approach to designing Linear Quadratic Gaussian (LQG) controllers for unknown linear systems subject to process and measurement noise.
This approach relies on solving semi-definite programs (SDPs) to construct a robust data-driven controller
that ensures robust global exponential stability for state estimation through a Kalman filter and input-to-state practical stability for control. However, a key assumption in their work is that noise-free state trajectories are available during offline data collection, which can be unrealistic in many practical scenarios where only noisy state data can be observed. In \cite{10124991}, data-driven controllers are learned for linear time-varying (LTV) systems using noisy input-state data. A convex optimization problem is formalized using data-dependent linear matrix inequalities (LMIs). However, a fundamental issue arises when the lumped uncertainty depends on the unknown system matrix, which is not explicitly available.  Therefore, the computed lumped uncertainty cannot be quantified, making it difficult to verify its properties or ensure that it satisfies the necessary constraints for robust control synthesis. 
Therefore, existing approaches still require access to noise-free state data during the offline phase or introduce additional estimation steps that reintroduce modeling assumptions. Hence, there remains a gap in the literature regarding direct data-driven control methods that can operate reliably when only noisy state measurements are available.

In  \cite{9992774}, a behavioral representation approach  is introduced to solve the LQG control problem directly using input-output data. This framework enables the formulation of the optimal LQG controller as a static feedback gain in behavioral space, which simplifies the subsequent data-driven design and implementation. Despite these advancements, a notable limitation of this approach is the requirement of expert-generated input-output data with specific rank conditions during the learning phase. Such assumptions may restrict practical applicability in scenarios with substantial measurement noise or limited data availability.

Together, these works provide various pathways to combine data-driven principles with classical optimal control tools in the presence of noisy or partial measurements. They highlight the potential of direct data-driven control approaches to maintain performance guarantees and stability without requiring extensive model identification steps. This paper builds on these ideas and focuses on new techniques for \emph{direct data-driven control using noisy state measurements}, aiming for both simplicity and robustness against the uncertainties encountered in practical applications.
This paper seeks to develop a direct data-driven control framework tailored for scenarios where one only has access to one trajectory of noisy state measurements. 
The main contribution of this work is to show that, even under these restrictive conditions, one can design effective feedback policies without requiring explicit identification of the underlying system. 
In contrast to prior studies \cite{DEPERSIS2021109548,ESMZAD2025112197} that presume noise-free state data, we introduce an approach that accommodates the realities of sensor noise in state accessibility.


\subsection{Main Contributions}

This paper introduces a novel direct data-driven control framework tailored for linear systems with noisy state measurements. The proposed methodology avoids explicit system identification and directly utilizes noisy input-output data to synthesize optimal stabilizing controllers. The main contributions are summarized below:

\begin{itemize}
\item \textbf{A Novel Uncertainty-aware Closed-loop Covariance Formalization:} We introduce a novel uncertainty-aware formulation of the closed-loop data-based covariance matrix that explicitly accounts for the propagation of process and measurement noise. Unlike previous works that neglect noise terms, our approach models these uncertainties directly in the controller synthesis, resulting in a more robust and reliable feedback gain design.
    \item \textbf{Theoretical Guarantee for Mean Square Stability (MSS) Preservation:} We provide a data-driven uncertain closed-loop system matrix in which its MSS satisfaction guarantees that the true nominal closed-loop system matrix also inherits the MSS property. This result (Theorem~\ref{th:stability}) provides the theoretical foundation for data-driven parameterization of the uncertain closed-loop matrix.

\item \textbf{Direct Data-Driven Robust Optimal Feedback Synthesis:} We develop a unified approach for synthesizing a stabilizing state feedback controller directly from noisy input-output data. First, we derive an LMI-based condition that ensures mean-square stability (MSS) of the closed-loop system. Building on this, we formulate a data-driven LQR problem solved via an SDP, which accounts for process and measurement noise and yields an optimal gain that ensures robust performance even from a single noisy trajectory (Theorems~\ref{th:ddstability} and~\ref{th:ddlqrnoisymeas}).


\end{itemize}

Collectively, these contributions offer a cohesive and tractable framework for controller design in realistic settings with sensor noise, and significantly extend the applicability of data-driven control methods to noise-corrupted environments.

\noindent
The paper is organized as follows. Section \ref{sec:prelim} reviews the system model and introduces the necessary notation and assumptions for both model-based and data-based settings. Section \ref{sec:mainresults} formulates the main theoretical results, including MSS preservation and data-driven stability conditions using noisy input-output data, and also extends the analysis to optimal control by presenting a novel data-driven LQR formulation under noisy state measurements and deriving the corresponding SDP. Section~\ref{sec:sim} presents simulation results comparing the proposed method with existing data-driven LQR techniques on two benchmark systems. Section~\ref{sec:con} concludes the paper and describes future research directions. \vspace{6pt}

\noindent \textit{Notation:} Let $\mathbb{R}^{m \times n }$ denote the real linear space for all real matrices with dimensions $m \times n$. $\mathbb{N}$ is the set of natural numbers. A positive (semi) definite matrix $P \in \mathbb{R}^{n \times n}$ is denoted by $P \succ 0,\: (P \succcurlyeq 0)$. The transpose of a matrix $Q$ is denoted by $Q^\top$.  For a matrix $A$, $\text{Tr}(A)$ is its trace. All random variables are assumed to be defined on a probability space $(\Omega,\mathcal{F},\mathbb{P})$, with $\Omega$ as the sample space, $\mathcal{F}$ as its associated $\sigma$-algebra and $\mathbb{P}$ as the probability measure. For a random variable $w: \Omega \longrightarrow \mathbb{R}^n$ defined in the probability space $(\Omega,\mathcal{F},\mathbb{P})$, with some abuse of notation, the statement $w \in \mathbb{R}^n$ is used to state a random vector with dimension $n$. For a random vector $w$, $\mathbb{E}[w]$ indicates its mathematical expectation.
{$\mathcal{N}(0, W)$ represents the probability density function of a Gaussian distribution with mean $0$ and covariance matrix $W.$ $\mathbb{N}$ is the set of natural numbers.

\section{Preliminaries}\label{sec:prelim}

\subsection{Model-Based System Description}
The dynamics of an LTI system can be presented as
\begin{subequations}
\label{eq:system}
\begin{align} 
    x_{k+1} &= A x_k + B u_k + w_k, \label{eq:syst} \\
    y_k &= x_k + \upsilon_k, \label{eq:meas}
\end{align}
\end{subequations}
where \( x_k \in \mathbb{R}^n \) is the system's state, \( u_k \in \mathbb{R}^m \) is the control input, and \( y_k \in \mathbb{R}^n \) is the measurement vector. The matrices \( A \in \mathbb{R}^{n \times n} \) and \( B \in \mathbb{R}^{n \times m} \) are the state transition and input matrices, respectively. The system's noise is governed by $\omega_k \sim \mathcal{N}(0, W)$ where 
$W \succ 0 \in \mathbb{R}^{n\times n}$ is its covariance. Furthermore, the measurement's noise is distributed according to $\upsilon_k \sim \mathcal{N}(0, V)$ where $V \succ 0 \in \mathbb{R}^{n\times n}$ is its covariance.
The control input is parameterized as
\begin{align}
    u_k = K y_k,  \label{eq:policy}
\end{align}
where \( K  \in \mathbb{R}^{m \times n}\) is a static feedback gain that will be designed directly using the collected input-output data from the system. The model-based closed-loop system is 
\begin{align}
    x_{k+1} = (A+BK) x_k + B K \upsilon_k + \omega_k.\label{eq:closedloop}
\end{align}
\begin{assumption}
The system's noise covariance $W$ and the measurement's noise covariance $V$ are known.
\end{assumption}
\begin{definition}[Mean Square Stability (MSS)]{\cite{LAI2023110685,1103840}}
For a given control policy $u_k=h(x_k)$, the system (\ref{eq:syst}) is called MSS if there exists a positive definite $\Sigma_{\infty} \in \mathbb{R}^{n \times n}$ such that $\lim_{k \xrightarrow{}\infty}{\norm{\mathbb{E}(x_k x_k^T)-\Sigma_{\infty}}}=0$ and $\lim_{k \xrightarrow{}\infty}{\norm{\mathbb{E}(x_k)}=0}.$ In this case, the control policy is called admissible. 
\end{definition} 
\subsection{Data-Based System Description}
Suppose that we have collected $N$ sequences of data as
\begin{subequations}
\label{eq:collected_data}
\begin{align} 
    U_0 &= \begin{bmatrix}
        u_0 & u_1 & \cdots & u_{N-1}
    \end{bmatrix}, \\
    Y_0 &= \begin{bmatrix}
        y_0 & y_1 & \cdots & y_{N-1}
    \end{bmatrix}, \\
    Y_1 &= \begin{bmatrix}
        y_1 & y_2 & \cdots & y_{N}
    \end{bmatrix},
\end{align}
\end{subequations}
where $U_0 \in \mathbb{R}^{m \times N}$, $Y_0 \in \mathbb{R}^{n \times N}$ and $Y_1 \in \mathbb{R}^{n \times N}$. The corresponding noise sequences $\Omega_0 \in \mathbb{R}^{n \times N}$, $\Upsilon_0 \in \mathbb{R}^{n \times N}$, and $\Upsilon_1 \in \mathbb{R}^{n \times N}$ are 
\begin{subequations}
\begin{align}
    \Omega_0 &= \begin{bmatrix}
        \omega_0 & \omega_1 & \cdots & \omega_{N-1}
    \end{bmatrix}, \\
    \Upsilon_0 &= \begin{bmatrix}
        \upsilon_0 & \upsilon_1 & \cdots & \upsilon_{N-1}
    \end{bmatrix}, \label{eq:Upsilon0}\\
    \Upsilon_1 &= \begin{bmatrix}
        \upsilon_1 & \upsilon_2 & \cdots & \upsilon_{N}
    \end{bmatrix},
\end{align}
\end{subequations}
in which we have no access and treat their columns as Gaussian random variables in our derivations. Also, we define the following intermediary uncorrupted state samples to which we have no access to 
\begin{subequations}
\label{eq:not_collected_state}
\begin{align} 
    X_0 &= \begin{bmatrix}
        x_0 & x_1 & \cdots & x_{N-1}
    \end{bmatrix}, \\
    X_1 &= \begin{bmatrix}
        x_1 & x_2 & \cdots & x_{N}
    \end{bmatrix}.
\end{align}
\end{subequations}

\begin{assumption}
The collected data matrix $D_0 = \begin{bmatrix}
    U_0^T & Y_0^T
\end{bmatrix}^T$ is sufficiently rich (informative), i.e., it has full row rank. That is,
\begin{align}
    rank(D_0) = m+n. \label{eq:rank}
\end{align}
\end{assumption}
\noindent Direct data-driven control aims to eliminate the need for an explicit system identification step and design the controller directly using the collected data (\ref{eq:collected_data}). One way is to parameterize the closed-loop system matrices using the collected data and then design a controller for it. This involves learning the closed-loop dynamics. To employ a parameterization of $K$, inspired by ~\cite{DEPERSIS2021109548}, for any $K$ there is a matrix $G \in \mathbb{R}^{N \times n}$  such that
\begin{subequations}
\begin{align} 
    \begin{bmatrix}
        K \\ I
    \end{bmatrix} &= D_0 G, \label{eq:KD0G} 
    \end{align}
\end{subequations}
\begin{lemma}
Consider the system \eqref{eq:syst}-\eqref{eq:meas} under Assumptions 1 and 2. Let the collected data be given by \eqref{eq:collected_data}. Then, under \eqref{eq:KD0G}, the data-based closed-loop system becomes
\begin{subequations}
\begin{align}
    A+BK &= (Y_1  - \Upsilon_1 - \Omega_0) G +  A \Upsilon_0 G, \label{eq:ABK}\\
    Y_0 G &= I, \label{eq:Y0GI}
\end{align}
\end{subequations}
\end{lemma}
\noindent \textit{Proof:}
Based on the system dynamics \eqref{eq:syst} and measurement model \eqref{eq:meas}, the collected data satisfies 
\begin{subequations}
\begin{align}
    X_1 &= A X_0 + B U_0 + \Omega_0, \label{eq:data-integrity} \\
    Y_0 &= X_0 + \Upsilon_0,
\end{align}
\end{subequations}
and, as a result, the dynamics of the measurements are
\begin{align}
    Y_1 &= A Y_0 + B U_0 - A \Upsilon_0 + \Upsilon_1 + \Omega_0. \label{eq:meas-dyn}
\end{align}
Multiplying \eqref{eq:meas-dyn} by $G$ from right yields
\begin{subequations}
\begin{align}
    Y_1 G &= A Y_0 G + B U_0 G - A \Upsilon_0 G + \Upsilon_1 G + \Omega_0 G. \label{eq:Y1G} 
\end{align}
\end{subequations}
Using the parametrization of \eqref{eq:KD0G} in this equation it yields \eqref{eq:ABK} and \eqref{eq:Y0GI}.

\begin{remark}
    The parameterization in Equation \eqref{eq:KD0G} follows the same structural framework as presented in \cite{DEPERSIS2021109548,ESMZAD2025112197}. However, it is specifically formulated for input-output data with noisy state measurements, as opposed to input-state data. Furthermore, the closed-loop matrix in Equation \eqref{eq:ABK} incorporates two additional terms, namely \( -\Upsilon_1 G \) and \( A \Upsilon_0 G \), which arise due to measurement noise. Notably, the closed-loop matrix also exhibits dependence on the unknown system matrix \( A \). This dependence renders the direct application of the parameterization in Equation \eqref{eq:ABK} infeasible for data-driven derivations. In \cite{10124991}, a bound on the lumped noise term \( A \Upsilon_0 - \Upsilon_1 - \Omega_0 \) was proposed to address this issue and enhance robustness. 
\end{remark}

\section{Main Results} \label{sec:mainresults}
This section first presents a data-based control solution with MSS guarantee for the linear system  \eqref{eq:syst}-\eqref{eq:meas}. It then extends the results to optimize a linear quadratic cost function over controllers with MSS guarantees. That is, it presents a data-based solution to the  Linear Quadratic Regulator (LQR) problem with MSS guarantee. 

\subsection{Direct Data-driven Stabilizing Solutions Using Noisy Measurements}
The parametrization in \eqref{eq:ABK} shows the dependence of $A+BK$ on the matrix $A$, which is unknown. This problem is usually handled in literature by assuming an upper bound on this uncertainty $A \Upsilon_0$~\cite{10124991}. In this paper, we propose to use the uncertain closed-loop matrix 
\begin{align}
    A+BK-A \Upsilon_0 G &= (Y_1  - \Upsilon_1 - \Omega_0) G , \label{eq:ABK_uncer}
\end{align}
instead of an unknown closed-loop matrix $A+BK$. To be a viable substitution and approach, we must show that the MSS of the uncertain closed-loop matrix
\begin{align}
    A+BK-A \Upsilon_0 G= A(I - \Upsilon_0 G) + B K \label{eq:ABK_uncertain}
\end{align}
implies the MSS of $A+BK.$ 
\begin{assumption}
The pair $(A, B)$ is unknown but stabilizable.
\end{assumption}

\begin{theorem}[MSS Preservation]\label{th:stability}
    Consider the system \eqref{eq:syst} under noisy measurements given by \eqref{eq:meas}. Let  the control policy be given by \eqref{eq:policy}. Let \( \Upsilon_0\) and \(G\) be given as \eqref{eq:Upsilon0} and \eqref{eq:KD0G}, respectively. Let Assumption 3 hold.  
    If the uncertain system matrix \eqref{eq:ABK_uncertain}
    provides MSS, then the actual closed-loop system matrix $A+BK$
    guarantees MSS as well.
\end{theorem}

\begin{proof}
    By assumption, the MSS condition holds for \( A (I - \Upsilon_0 G) + B K \), which implies the existence of a positive definite matrix \( P \) such that  
    \[
        \mathbb{E} \left[ (A (I - \Upsilon_0 G) +  B K)^T P (A (I - \Upsilon_0 G) +  B K) \right] - P \preceq 0.
    \]
Expanding the expectation, we obtain  
    \begin{align}
        &\mathbb{E} \left[ (I - \Upsilon_0 G)^T A^T P A (I - \Upsilon_0 G) + 2 (I - \Upsilon_0 G)^T A^T P B K \right] \notag \\
        &+ K B^T P B K - P \preceq 0.
    \end{align}
    Rearranging terms leads to 
    \begin{align}
        &A^T P A + \mathbb{E} \left[ G \Upsilon_0^T A^T P A \Upsilon_0 G \right] \notag \\
        &+ 2 A^T P B K + K B^T P B K - P \preceq 0.
    \end{align}
    Using the inequality above, one can conclude that  
    \[
        (A + B K)^T P (A + B K) - P \preceq -G \mathbb{E} \left[ \Upsilon_0^T A P A \Upsilon_0 \right] G.
    \]
    Since the unknown right-hand side matrix is negative semi-definite, it follows that \( A + B K \) satisfies the MSS condition. This completes the proof.
\end{proof}

\noindent Theorem \ref{th:stability} establishes a foundation for deriving a data-driven formulation without the need to explicitly account for the uncertain term \( A\Upsilon_0 G \) in the data-driven representation of the closed-loop matrix \eqref{eq:ABK}. Specifically, it ensures that the right-hand side of \eqref{eq:ABK_uncer} can serve as a valid data-driven realization of the nominal closed-loop system \( A + BK \). The corresponding stability analysis in the data-driven setting is formalized in Theorem \ref{th:ddstability}.

\begin{theorem}[Data-driven MSS]\label{th:ddstability}
    Consider the system (\ref{eq:system}) and let the data collected from this system be given as (\ref{eq:collected_data}), and Assumptions 1-3 hold. Let there exists a $\gamma >0$ for which the data-driven control gain $K = U_0 G$ with $G=FY^{-1}$ and $Y=P^{-1}$ solves the following program
\begin{subequations}
\label{eq:directRobustlmi}
\begin{align}
 &\begin{bmatrix}
        -Y & (Y_1 F)^T & F^T \\
        * & -Y  & 0\\
        * & * & -\gamma I
        \end{bmatrix} \preceq 0. \label{eq:lmiDD} \\
        \quad & Y_0F=Y, \\
        \quad &Tr\left((V+W)^{-1}Y\right)-\gamma n^2\ge 0, \label{eq:tracePW} \\
        \quad & \gamma > 0,
\end{align}
\end{subequations}
Then, the closed-loop system is MSS.
\end{theorem}
\begin{proof}
  Based on Theorem 1, the MSS of the uncertain data-driven representation outlined in (\ref{eq:ABK_uncer}) results in the MSS of the closed-loop system. On the other hand, (\ref{eq:ABK_uncer}) is MSS if there exists $P >0$ such that
\begin{align}
&\mathbb{E}\bigl[ G^\top (Y_1  - \Upsilon_1 - \Omega_0)^\top P (Y_1  - \Upsilon_1 - \Omega_0) G\bigr] -   P\preceq 0.  
\end{align}
By treating columns of $\Upsilon_1$ and $\Omega_0$ as Gaussian random variables and taking the expectation concerning those random variables, one has
\begin{align}
&G^\top Y_1^\top P Y_1 G + G^\top Tr(PW + PV) G - P \preceq 0, \label{eq:P_lmi_data}
\end{align}
where
\[
 G^\top \mathbb{E}\bigl[  \Upsilon_1^\top P \Upsilon_1 +  \Omega_0^\top P \Omega_0 \bigr] G =  G^\top Tr(PW + PV) G.
\]
By pre- and post-multiplying both sides of the inequality (\ref{eq:P_lmi_data}) by $P^{-1}$, using the Schur-complement twice, and the following inequalities 
\begin{align}
    &\frac{1}{Tr((V+W)^{-1}Y)} < \frac{Tr(P(W+V))}{n^2} \label{eq:PYineq}, \\
    &Tr(P (W+V)) \leq \frac{1}{\gamma}
\end{align}
the data-driven MSS LMI (\ref{eq:lmiDD}) and inequality (\ref{eq:tracePW}) are resulted, where $\gamma > 0$ is a decision variable. The term $G^\top Tr(PW + PV) G$ in \eqref{eq:P_lmi_data} brings robustness into the design of the data-driven controller. 
This completes the proof. 
\end{proof}

\noindent Theorem \ref{th:ddstability} presents a data-driven approach for ensuring the MSS of the closed-loop system without requiring explicit system identification. By leveraging the collected input-output data \eqref{eq:collected_data} and formulating the control gain \( K \) directly through the optimization program \eqref{eq:directRobustlmi}, the stability condition is enforced purely in terms of data matrices.
The constraint \eqref{eq:lmiDD} defines an LMI that provides a robust stabilization guarantee under measurement and process noise, while the trace condition \eqref{eq:tracePW} incorporates an uncertainty-aware design criterion. The optimization formulation ensures that the data-driven controller accounts for the stochastic nature of the noise terms \( W \) and \( V \), enhancing robustness in real-world scenarios. 
In general, Theorem \ref{th:ddstability} provides a principled framework for designing stabilizing controllers directly from the noisy input-output data, ensuring the MSS of the closed-loop system in a purely data-driven manner. This result paves the way for practical implementation in systems where obtaining an accurate model is challenging or infeasible.

\subsection{Direct Data-driven LQR from Noisy State Measurements}
Having established the MSS of the data-driven closed-loop system in Theorem \ref{th:ddstability}, we now extend our analysis to the optimal control design. Specifically, we consider the LQR problem in the presence of noisy measurements. The goal is to design an optimal data-driven controller that minimizes the quadratic cost function
 \begin{align} 
J(K)&=\lim_{\tau\to\infty} \frac{1}{\tau} \sum_{k=0}^{\tau}{\mathbb{E}\left[x_k^\top Q x_k + u_k^\top R u_k\right]}, \label{eq:lqrcost}
\end{align}
while ensuring robustness against measurement and process noise.
\begin{assumption}
For the system \eqref{eq:syst} and the cost function (3), the pair $(A, \sqrt{Q})$ is detectable. 
\end{assumption}

In existing direct data-driven control methods for the LQR problem \cite{9705109,DEPERSIS2021109548,ESMZAD2025112197}, the Bellman inequality or the controllability Grammian (covariance) inequality, which are expressed by the closed-loop dynamics $A+BK$, are solved using data by characterizing $A+BK$ based on data and a decision variable. That is, in these approaches that rely on access to noise-free state measurements, it is not necessary to identify the input matrix \( B \), and controller synthesis can be performed directly from input-state data. However, when only noisy state measurements are available, as in this paper, the controller design relies on a data-driven representation of the input matrix \( B \) as well as the closed-loop dynamics $A+BK$ (see equation \eqref{eq:closedloop}). A key challenge arises because the standard parameterization of \( B \) from data depends on the unknown matrix \( A \), making it infeasible to use directly. To address this, we propose an uncertainty-aware reformulation of the closed-loop matrix that removes the explicit dependence on \( A \), thereby enabling a tractable synthesis procedure. 

Based on equation \eqref{eq:closedloop}, to find a closed-loop data-driven parametrization, we need to parameterize the input matrix $B$ in terms of the noisy input-output data. To this end, consider the following parameterization 
\begin{align}
    \begin{bmatrix}
        I \\ 0
    \end{bmatrix} &= D_0 M. \label{eq:MD0}
\end{align}
where $M \in \mathbb{R}^{N \times n}.$ Now, if we multiply both sides of \eqref{eq:meas-dyn} with $M$ from right, we have 
\begin{align}
      Y_1 M &= A Y_0 M + B U_0 M - A \Upsilon_0 M + \Upsilon_1 M + \Omega_0 M. \label{eq:Y1M} 
\end{align}
Now, using the parameterization in \eqref{eq:MD0}, one has
\begin{align}
    B &= (Y_1 - \Upsilon_1 - \Omega_0 )M + A \Upsilon_0 M, \label{eq:B}
\end{align}
which provides an uncertain realization of the input matrix $B$ using the collected data.
This equation also depends on the matrix $A$, which is unknown. To alleviate this issue, as before, we consider the following uncertain input matrix
\begin{align}
    B-  A \Upsilon_0 M &= (Y_1 - \Upsilon_1 - \Omega_0 )M , \label{eq:Buncertain}
\end{align}
instead of the nominal matrix $B.$
This reformulation enables the derivation of a data-driven closed-loop system as shown next, where the uncertain dynamics are defined purely in terms of measured data and decision variables $G$ and $M$ as
\begin{align}
    x_{k+1} &= ( A+BK-A \Upsilon_0 G) x_k + (B - A \Upsilon_0 M)\upsilon_k + \omega_k,\nonumber \\
    x_{k+1} &= (Y_1 - \Upsilon_1 - \Omega_0 )(Gx_k + M K \upsilon_k) + \omega_k. \label{eq:dataK}
\end{align}
The steady state covariance for the uncertain dynamics in \eqref{eq:dataK}
\begin{equation} 
\Sigma = {\lim_{k\to\infty}} \mathbb{E} [x_{k+1} \,\, x_{k+1}^T] = {\lim_{k\to\infty}} \mathbb{E} [x_k \,\, x_k^T], \label{eq:sscov} 
\end{equation}
can be computed as
\begin{align}
& \Sigma = \mathbb{E} [((Y_1 -\Upsilon_1 - \Omega_0) (Gx_k + M K \upsilon_k)+\omega_{k})  (*)^T] \nonumber \\ 
& \Sigma = Y_1 G \Sigma G^\top Y_1^\top + Tr(G \Sigma G^\top)(W + V) + W \nonumber \\
& + Y_1 M K V K^\top M^\top Y_1^\top + Tr(M K V K^\top M^\top)(V+W), \label{eq:sigma}
\end{align}
where the expectation is taken over all random variables $x_k, \upsilon_k,$ and $\omega_k$ and unknown noise matrices $\Upsilon_1,\Omega_0$, which we treat their columns as random variables. 
Given this structure, our goal is to compute the optimal control gain \( K \) that minimizes the standard infinite-horizon quadratic cost \eqref{eq:lqrcost} while ensuring MSS and robustness against measurement noise. The following theorem provides a data-driven formulation for the LQR problem under noisy state measurements, using a convex SDP program to determine the optimal feedback gain.


\begin{remark} \label{rem:cov1}
    The data-based closed-loop covariance expression in~\eqref{eq:sigma} includes an additional term, $\mathrm{Tr}(G \Sigma_k G^\top)W$, which is absent in the formulation presented in~\cite{DEPERSIS2021109548}. This discrepancy arises because~\cite{DEPERSIS2021109548} neglects the contribution of the term $\Omega_0 G$ in the parameterization of the closed-loop matrix $A + BK$. In contrast, our approach explicitly models $\Omega_0 G$ as a random variable, allowing us to capture the effect of process noise propagation through the gain matrix $G$. Moreover, the presence of measurement noise introduces further terms into the covariance structure. This enhanced, uncertainty-aware covariance formulation contributes to the robustness of the synthesized feedback gain.
\end{remark}

\begin{theorem}[Data-driven LQR Using Noisy Measurements]\label{th:ddlqrnoisymeas}
 Consider the system (\ref{eq:syst}) with closed-loop data-based parameterization (\ref{eq:dataK}) and the steady state covariance \eqref{eq:sscov}. Let Assumptions 1-4 hold and 
 $x_0$ be a given initial condition. Then, the optimal state feedback matrix $K=U_0 G$ that minimizes the LQR cost \eqref{eq:lqrcost} 
 can be computed through the following SDP
\begin{subequations}
 \label{eq:ddlqr}
\begin{align}
    \min_{\beta>0, M, \Sigma, H, S, Z, E, F } \quad  {\beta}
\end{align}
\begin{align}
\textrm{s.t.} &\quad  Y_1 (H+ Z) Y_1^\top + Tr(H+Z) (W+V) -\Sigma+W \preceq 0 \label{eq:Sigma} \\
     & \quad  \begin{bmatrix}
         E & F\\
         * & S
     \end{bmatrix} \succeq 0, \label{eq:E}\\
     & \quad \begin{bmatrix}
        H & F \\ F^\top & \Sigma
    \end{bmatrix} \succeq 0, \label{eq:H}\\
    & \quad \begin{bmatrix}
        S & \Sigma \\ * & V
    \end{bmatrix} \succeq 0, \label{eq:S}\\
    & \quad U_0 (Z - E) U_0^\top = 0_{m \times m}, \label{eq:Z} \\
     & \quad  Y_0 F = \Sigma, U_0 M = I_m, Y_0 M = 0_{n \times m}, \label{eq:y0u0gm}\\
     & \quad  Tr(Q \Sigma)+Tr(R U_0 H U_0^\top) + Tr(RU_0 E U_0^\top) \leq \beta, \label{eq:Beta}
\end{align}
\end{subequations}
where $G=F \Sigma^{-1}$. Then,  the gain $K=U_0F\Sigma^{-1}$ guarantees MSS of the actual closed-loop system \eqref{eq:closedloop}.
\end{theorem}
\begin{proof}
    By using the characteristics of the trace and the steady-state covariance definition \eqref{eq:sscov}, the cost \eqref{eq:lqrcost} can be simplified as 
    \begin{align}
        J(K) = Tr(Q \Sigma) + Tr(R K \Sigma K^\top) + Tr(R K V K^\top). \label{eq:lqrcost_s}
    \end{align}
    Now, consider the multiplication inside the second trace term 
    \begin{align}
         K \Sigma K^\top &= U_0 G \Sigma G^\top U_0^\top = U_0 F \Sigma^{-1} F^\top U_0^\top, \nonumber \\
         & = U_0 H U_0^\top, \label{eq:u0hu0t}
    \end{align}
    where we defined a new variable $H=F \Sigma^{-1} F^\top$ which is equivalent to the LMI \eqref{eq:H}. Also, the matrix inside the third trace term can be simplified as  
    \begin{align}
        K V K^\top & = U_0 F \Sigma^{-1} V \Sigma^{-1} F^\top U_0^\top \nonumber \\
        & = U_0 E U_0^\top, \label{eq:u0eu0t}
    \end{align}
    where we used $E=F S^{-1} F^\top$ and $S^{-1}=\Sigma^{-1} V \Sigma^{-1}$ which are equivalents to \eqref{eq:E} and \eqref{eq:S} respectively. By substituting \eqref{eq:u0hu0t} and \eqref{eq:u0eu0t} into \eqref{eq:lqrcost_s}, the left-hand side of \eqref{eq:Beta} will result.
    Then, consider the following term from the steady state covariance relation \eqref{eq:sscov}
    \begin{align}
        Z &= M K V K^\top M^\top \nonumber \\
        Z &= M U_0 E U_0^\top M^\top, \label{eq:ZE}
    \end{align}
    and assuming that $U_0 M = I$ is satisfied, we multiply \eqref{eq:ZE} from left and right by $U_0$ and $U_0^\top$ respectively, we get \eqref{eq:Z}. By substituting $H$ and $Z$ into \eqref{eq:sigma}, the LMI condition \eqref{eq:Sigma} will result. The conditions in \eqref{eq:y0u0gm} come from the parameterizations in \eqref{eq:KD0G} and \eqref{eq:MD0}. As a result of the optimization problem \eqref{eq:ddlqr}, the gain $K=U_0F\Sigma^{-1}$ stabilizes the uncertain data-driven parameterized system \eqref{eq:dataK}. Also, according to the analysis in Theorems \ref{th:stability} and \ref{th:ddstability}, this gain provides MSS for the actual closed-loop system \eqref{eq:closedloop}. 
    \end{proof}

\noindent
Theorem \ref{th:ddlqrnoisymeas} provides a principled approach to solving the LQR problem directly from noisy input-output data, without requiring explicit knowledge of the system matrices \( A \) and \( B \). By formulating the control objective as an SDP, the result guarantees not only optimality with respect to the quadratic cost but also the MSS of the closed-loop system. The controller gain \( K = U_0 F \Sigma^{-1} \) is constructed purely from data, and the formulation inherently incorporates robustness against both process and measurement noise through the covariance terms. 

\begin{remark}
    In the derivation of condition \eqref{eq:Z}, we assumed that the relation \( U_0 M = I_m \) holds prior to solving the optimization problem. This assumption facilitates the simplification of the nonlinear term in \eqref{eq:ZE}. Consequently, before proceeding with the semidefinite program \eqref{eq:ddlqr}, one may verify the feasibility of the linear constraints \( U_0 M = I_m \) and \( Y_0 M = 0_{n \times m} \). If a solution for \( M \) exists, the assumption used in \eqref{eq:ZE} is justified, and the optimization problem \eqref{eq:ddlqr} can be solved accordingly.
Furthermore, since \( Z = E \) is a trivial solution that satisfies \eqref{eq:Z}, one may optionally substitute \( Z \) with \( E \) in both \eqref{eq:Z} and \eqref{eq:Sigma}, potentially simplifying the optimization formulation. This observation allows the overall procedure to be modular: first, validate the feasibility of \( M \), and then proceed to solve the remaining part of the optimization problem.
\end{remark}
\section{Simulation Results} \label{sec:sim}
This section compares the presented direct data-driven from noisy state measurements LQR (DDNSMLQR) presented in Theorem \ref{th:ddlqrnoisymeas}, with model-based LQR using noisy state measurements (NSMLQR), the direct data-driven low complexity LQR (DDLCLQR) approach ~\cite{DEPERSIS2021109548}, and the direct data-driven dynamic programming LQR (DDDPLQR) method \cite{ESMZAD2025112197}. Specifically, we compare the number of times a controller is successfully found (i.e., a feasible solution exists to the resulting optimization), the number of times for which the closed-loop eigenvalues were all inside the unit circle, the average controller gains $K$ obtained from each method and their norms with respect to the model-based LQR control gain.
We used CVXPY~\cite{diamond2016cvxpy,agrawal2018rewriting} for modeling the convex optimization problems and MOSEK~\cite{mosek} as the solver.

The tuning parameters used in each of the methods are as follows: $\alpha\geq0$ is the robustness parameter in \cite{DEPERSIS2021109548} and $\lambda=0.99$ is the discount factor in \cite{ESMZAD2025112197}.
Also, we treat $W$ and $V$ as design parameters in case we have no access to the actual covariances. 
To do a fair comparison, all the methods are assumed to have access to the same data set. Also, we use the noisy measurements $Y_0$ and $Y_1$ instead of $X_0$ and $X_1$ respectively to DDLCLQR and DDDPLQR methods. 
We consider two examples. The first one is a stable active car suspension system, and the second one is the unstable rotary inverted pendulum system. 

\subsection*{Example 1: Active Car Suspension}
Consider the following two degrees of freedom quarter-car suspension system~\cite{rajamani2011vehicle,roh1999stochastic}
\begin{align}
    \dot{x}(t)&=\begin{bmatrix}
        0 &1 &0 & -1 \\
        -\frac{k_s}{m} & -\frac{b_s}{m_s} & 0 & \frac{b_s}{m_s} \\
        0 & 0 &0& 1\\
        \frac{k_s}{m_u} & \frac{b_s}{m_u} & -\frac{k_t}{m_u} & -\frac{b_s}{m_u}
    \end{bmatrix} x(t) \nonumber \\
    & \quad + \begin{bmatrix}
        0\\ \frac{1}{m_s} \\ 0 \\ -\frac{1}{m_u}
    \end{bmatrix} u(t) + \omega(t),
\end{align}
where $x=\begin{bmatrix}
    x_1 & x_2 & x_3 & x_4
\end{bmatrix}^T$. The states $x_1$, $x_2$, $x_3$, and $x_4$ show the suspension deflection, absolute velocity of sprung mass, tire deflection, and absolute velocity of unsprung mass, respectively. Also, $m_s=240~kg$, $m_u=36~kg$, $b_s=980~\frac{N.s}{m}$, $k_s=16000~\frac{N}{m}$, and $k_t=160000~\frac{N}{m}$ represent the quarter-car equivalent of vehicle body mass, the unsprung mass due to axle and tire, suspension damper, suspension spring, and vertical stiffness spring respectively. The sampling time used for data collection is $T_s=0.05~s.$ The noise vector $\omega(t)$ shows uncertainties in system dynamics due to modeling errors and random vertical road displacement, which is assumed to have a Gaussian distribution with zero mean and the covariance matrix $W=10^{-7}I$.
The control input $u(t)$ represents the active force actuator of the suspension system. The measurement noise covariance is $V=2 \times 10^{-5} I.$ The following input is applied to the system
$
    u_k = 400\eta_k,
$
where $\eta_k$ is a zero-mean and unit-variance Gaussian noise. We have generated $150$ sets of input-(noisy )output data with $N=10$ and designed $N_K=150$ controllers using this data. For each of the designed controllers, $N_S=50$ simulations of $N_P=20$ sample points are performed. The initial condition is chosen as $\begin{bmatrix}
    0.3 & -4 & 0.1 & -1
\end{bmatrix}^T$. 
The optimal model-based gain $K^*$ for $Q=diag(10000, 1, 1, 1)$ and $R=0.000001$ is in the first row of Table \ref{tab:suspension}. The \textit{Stable Eigenvalues} column reports the eigenvalues of the actual closed-loop system when using the controller gain derived from noisy data. The \textit{Success Count} column indicates the number of instances in which the optimization problem was successfully solved based on the collected data.
In the first scenario (as depicted in Figure \ref{fig:susp-1}) we choose $\alpha=0$ for the DDLCLQR method. Simulation results in the second to fourth rows of Table \ref{tab:suspension} show that DDNSMLQR performs well with a high success rate and the least gain error. However, DDDPLQR and DDLCLQR exhibit a significant deviation in the feedback gain norm, showing weaker performance, which can be attributed to the lack of measurement noise covariance knowledge in them.

For $\alpha=100$ in DDLCLQR, the method remains stable with a success count and stable designs of 150, but the deviation from $K^*$ increases, showing neither robustness nor performance (Figure \ref{fig:susp-2}).

If we treat $W$ as a design parameter and set it to $W=10^{-4}I$ for DDDPLQR to account for more uncertainties due to measurement noise and set the parameter $\alpha=1$ for DDLCLQR which accounts for robustness against uncertainties, then DDDPLQR improves significantly with a success count of only 142, and DDLCLQR exhibits robust performance (Figure \ref{fig:susp-3}). Although setting appropriate values for $\alpha$ and $W$ can enhance the robust performance of DDLCLQR and DDDPLQR methods, there is no proper approach to calculate them.

\begin{figure}[]
\centerline{\includegraphics[width=0.9\columnwidth]{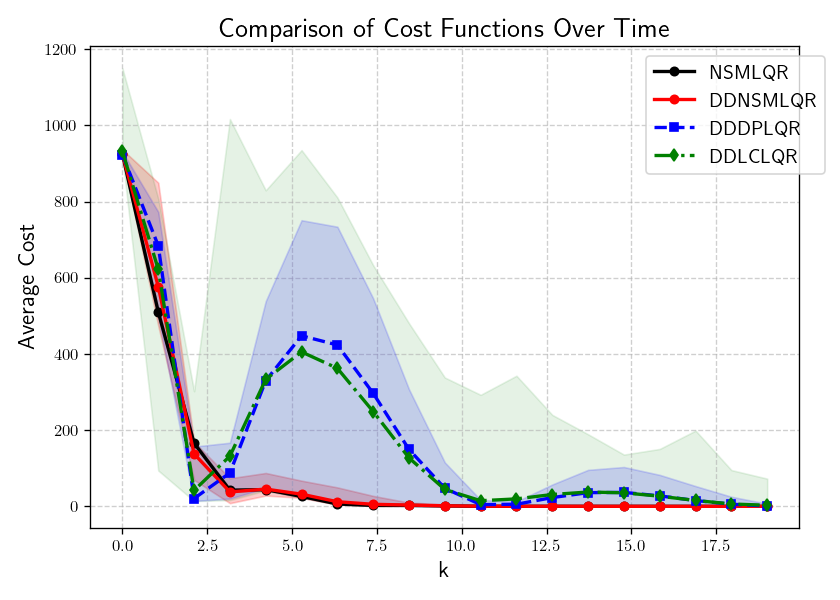}}
    \caption{Average closed-loop cost comparison for all methods on Example 1 (Active Suspension System) with DDLCLQR regularization parameter \(\alpha = 0\). DDNSMLQR achieves the lowest cost, while other methods suffer from poor stability or suboptimal gains.}

\label{fig:susp-1}
\end{figure}
\begin{figure}[]
\centerline{\includegraphics[width=0.9\columnwidth]{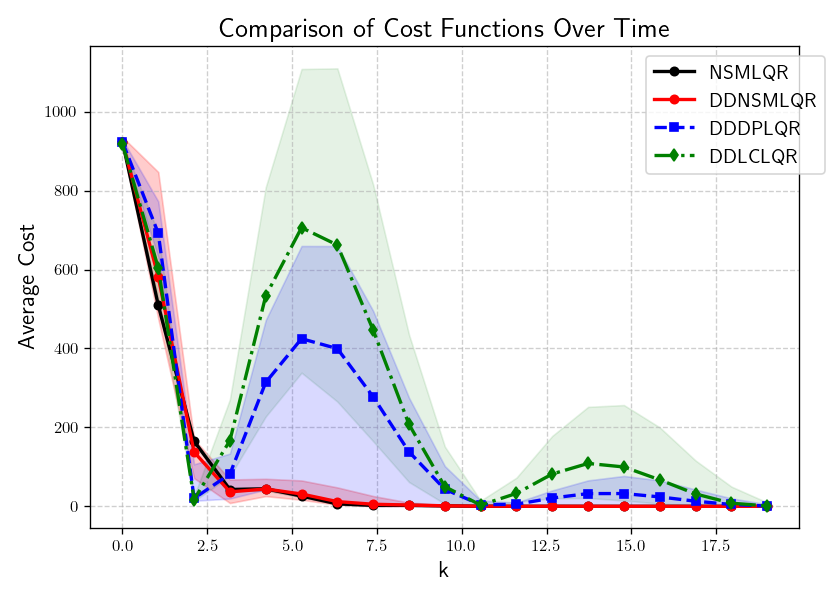}}
    \caption{Average closed-loop cost for all methods on Example 1 with DDLCLQR using \(\alpha = 100\). The increased regularization improves stability for DDLCLQR but leads to higher cost compared to DDNSMLQR.}
\label{fig:susp-2}
\end{figure}
\begin{figure}[]
\centerline{\includegraphics[width=0.9\columnwidth]{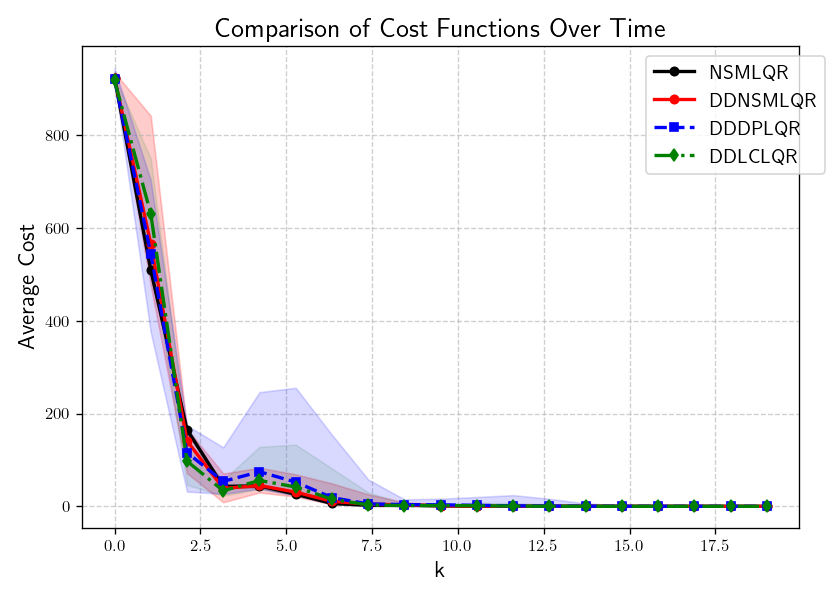}}
    \caption{Average cost comparison on Example 1 with DDDPLQR using \(W = 10^{-4}I\) and DDLCLQR using \(\alpha = 1\). Tuning these parameters improves the performance of both methods.}
\label{fig:susp-3}
\end{figure}

\begin{table*}[h]
    \centering
    \caption{Performance Comparison of Data-Driven LQR Methods Using Noisy State Measurements for Example 1}
    \label{tab:suspension}
    \renewcommand{\arraystretch}{1.3}
    \resizebox{\textwidth}{!}{%
    \begin{tabular}{clcccc}
        \toprule
        \midrule
        & \textbf{Method} & \textbf{Success Count} & \textbf{Stable Eigenvalues} & $\mathbf{\|K - K^*\|}$ & \textbf{Average Gain Matrix $K$} \\
        \midrule
        \midrule       
        & NSMLQR  & ---  & ---  & --- & \multicolumn{1}{l}{$\begin{bmatrix} -35829.36 & -3068.65 & -43378.42 & -130.83 \end{bmatrix}$} \rule{0pt}{15pt} \\
        & DDNSMLQR (Theorem \ref{th:ddlqrnoisymeas})  & 150  & 150  & 9969.88  & \multicolumn{1}{l}{$\begin{bmatrix} -28630.28 &  -2758.74 & -36488.24 &  -93.65 \end{bmatrix}$} \rule{0pt}{15pt} \\
        & DDDPLQR~\cite{ESMZAD2025112197}   & 73  & 73  & 54840.50  & \multicolumn{1}{l}{$\begin{bmatrix} -1497.58 & -729.51 &  -678.14 &   26.80 \end{bmatrix}$} \rule{0pt}{15pt} \\
        & DDLCLQR~\cite{DEPERSIS2021109548} $(\alpha=0)$   & 150  & 150  & 46437.26  & \multicolumn{1}{l}{$\begin{bmatrix} -6361.79 & -893.031 & -7554.82 & 4.62 \end{bmatrix}$} \rule{0pt}{15pt} \\
        \midrule
        & DDLCLQR~\cite{DEPERSIS2021109548} $(\alpha=100)$   & 150  & 150  & 54446.96  & \multicolumn{1}{l}{$\begin{bmatrix} -1341.01 & -398.45 & -1332.27 & 0.24\end{bmatrix}$} \rule{0pt}{15pt}\\
        \midrule
       
        & DDDPLQR~\cite{ESMZAD2025112197}  $(W=10^{-4}I)$  & 142   & 142   & 12559.17  & \multicolumn{1}{l}{$\begin{bmatrix} -27019.24 &  -2522.63 & -34444.44 & -106.96 \end{bmatrix}$} \rule{0pt}{15pt} \\
        & DDLCLQR~\cite{DEPERSIS2021109548} $(\alpha=1)$  & 150  & 150  & 20240.72   & \multicolumn{1}{l}{$\begin{bmatrix} -24525.53 &  -2606.93 & -26594.61 &
           -96.45\end{bmatrix}$} \rule{0pt}{15pt}\\
        \midrule
        \bottomrule
    \end{tabular}%
    }
\end{table*}

\subsection*{Example 2: Rotary Inverted Pendulum}
We also conducted simulations on a widely used benchmark platform \textit{Quanser rotary pendulum (model QUBE-Servo 2)}
\footnote{\url{https://www.quanser.com/products/qube-servo-2/}} 
shown in Figure \ref{fig:quanser}. The system dynamics are represented in discrete time with a sampling interval of \( T_s = 0.05 \) seconds. The state vector is defined as $x = \begin{bmatrix} \theta & \delta & \dot{\theta}   & \dot{\delta} \end{bmatrix}^T$,
where \( \theta \) (rad) is the rotary arm angle and \( \delta \) (rad) is the pendulum angle. 
The control input \( u  \) (V) is the voltage applied to the motor. Here, we considered linearized dynamics around its upright equilibrium point.
The control strategy aims to stabilize the pendulum in its upright equilibrium while ensuring smooth control effort. Both the natural gradient control method and LQR are employed to regulate the trajectory of the state vector. The discrete-time state-space matrices used in the simulation are given by
\begin{align*}
A &=
\begin{bmatrix}
1.0 &0.63 & 0.037 & 0.0023 \\
0 & 1.31 & -0.013 & 0.054\\
0 & 6.23 & 0.52 & 0.14 \\
0 & 1.24 & -0.49 & 1.27
\end{bmatrix}, 
B &=
\begin{bmatrix}
0.053 \\
0.054 \\
2.00  \\
2.05
\end{bmatrix}.   
\end{align*}
The sampling time used for data collection is $T_s=0.05~s.$ The state noise covariance is $W=10^{-6}I$.
The measurement noise covariance is $V=2 \times 10^{-4} I.$ The following input is used for data collection purposes
\[
    u_k = \begin{bmatrix}
    0.87 & -16.65 &    0.73 &  -1.30
\end{bmatrix}x_k + 0.45 \eta_k.
\] We have generated $150$ sets of input-(noisy )output data with $N=10$ and designed controllers using them. For each of the designed controllers, $N_S=50$ simulations of $N_P=20$ sample points are performed. The initial condition is chosen as $\begin{bmatrix}
    0.2 & -0.2 & 0 & 0
\end{bmatrix}^T$. 
The optimal model-based gain $K^*$ for $Q=diag(1, 100, 1, 100)$ and $R=10$ is in the first row of Table \ref{tab:inverted}.

\noindent
Table~\ref{tab:inverted} presents a comparative evaluation of several data-driven LQR controllers designed under noisy state measurements for Example 2. The reference controller, NSMLQR, is a model-based implementation of the noise-sensitive mean-square optimal LQR and serves as a benchmark for evaluating the performance of data-driven approaches.

Among all data-driven methods, the proposed DDNSMLQR controller (Theorem~\ref{th:ddlqrnoisymeas}) achieves the most favorable trade-off between accuracy and robustness. It successfully stabilizes the system in 121 out of 150 trials and achieves a low gain error of \( \|K - K^*\| = 0.78 \), closely matching the optimal reference gain. The resulting gain matrix is also numerically similar to that of NSMLQR, indicating high fidelity in replicating the model-based controller using only noisy data (Figure \ref{fig:quan2}).

In contrast, the DDDPLQR and DDLCLQR methods show significantly degraded performance under default or poorly tuned parameters. For example, both methods achieve only 2 stable eigenvalue configurations without regularization (\(\alpha=0\)) or proper noise scaling. While tuning parameters such as \(W=10^{-3}I\) for DDDPLQR and \(\alpha=0.01\) for DDLCLQR improve stability (up to 77 and 66 stable cases, respectively), they still fall short in gain accuracy compared to DDNSMLQR (please Figure \ref{fig:quan4}).

Moreover, DDLCLQR with high regularization (\(\alpha=10\)) attains full stability but at the cost of overestimated gains and larger error (\(\|K - K^*\| = 2.26\)). Similarly, DDDPLQR becomes more stable with higher noise levels (\(W=10^{-3}I\)), but the resulting gains remain less accurate.

The DDLCLQR controller with \(\alpha = 0.1\) achieves full stability across all trials (150/150) and exhibits the lowest gain error (\(\|K - K^*\| = 0.42\)) among the DDLCLQR variants, showing that moderate regularization effectively balances robustness and accuracy as depicted in Figure \ref{fig:quan3}.

Overall, DDNSMLQR offers superior consistency and precision in both stability and controller accuracy. Unlike other methods that require careful parameter tuning, DDNSMLQR directly incorporates noise modeling into its synthesis process, leading to better robustness and more reliable control performance in data-driven settings.

\begin{figure}[]
\centerline{\includegraphics[width=0.3\columnwidth]{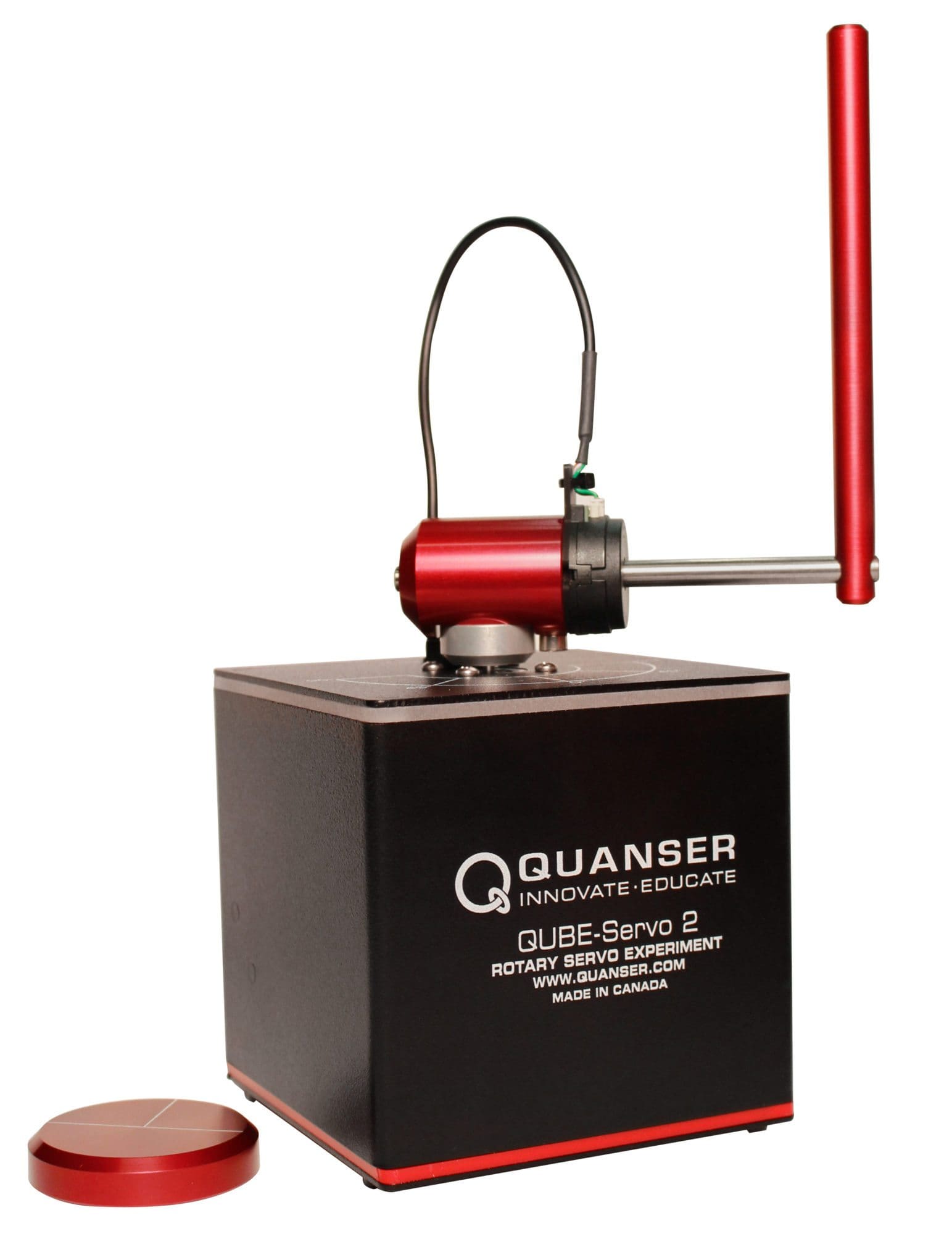}}
    \caption{Quanser Qube-Servo 2 platform used for simulating the rotary inverted pendulum dynamics in Example 2.}
\label{fig:quanser}
\end{figure}

\begin{figure}[]
\centerline{\includegraphics[width=0.9\columnwidth]{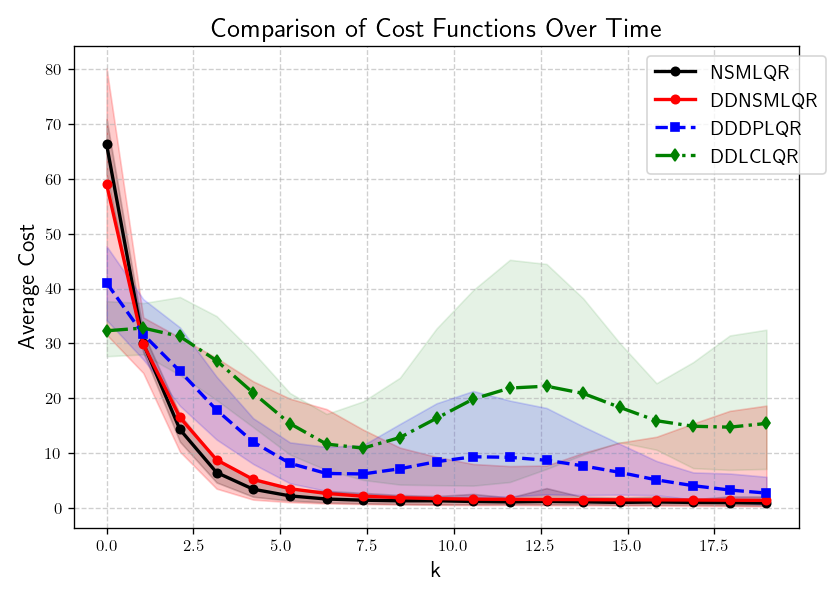}}
    \caption{Average closed-loop cost comparison for all methods on Example 2 with DDLCLQR regularization parameter \(\alpha = 0\). The absence of regularization results in instability for several approaches, while DDNSMLQR maintains robust performance.}
\label{fig:quan1}
\end{figure}

\begin{figure}[]
\centerline{\includegraphics[width=0.9\columnwidth]{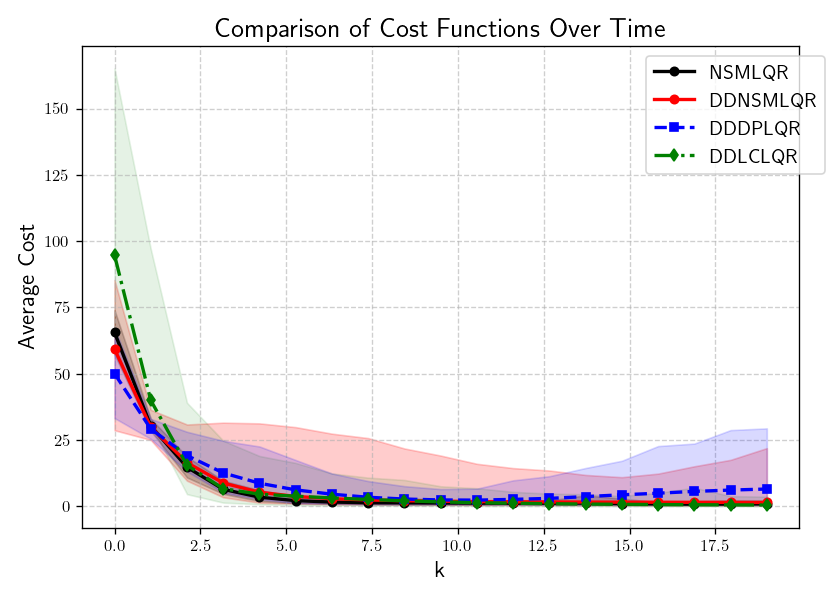}}
    \caption{Average cost for all methods on Example 2 with DDLCLQR using \(\alpha = 10\) and with DDDPLQR using \(W = 10^{-5}I\). Although DDLCLQR achieves full stability, its aggressive regularization leads to higher costs. DDLCLQR still suffers poor stability.}
\label{fig:quan2}
\end{figure}

\begin{figure}[]
\centerline{\includegraphics[width=0.9\columnwidth]{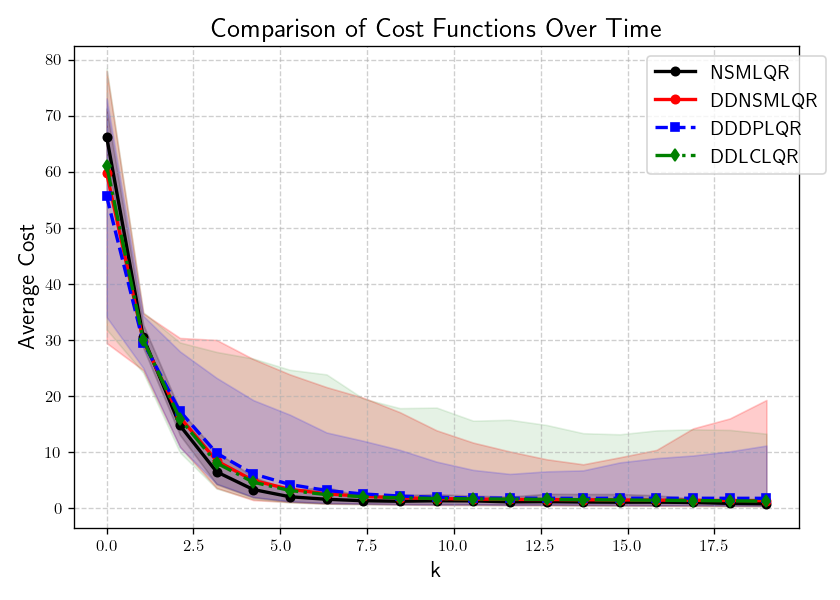}}
    \caption{Average cost comparison for all controllers on Example 2 with DDDPLQR using noise scaling \(W = 10^{-4}I\) and DDLCLQR with \(\alpha = 0.1\). Moderate tuning improves the performance of both approaches.}

\label{fig:quan3}
\end{figure}

\begin{figure}[]
\centerline{\includegraphics[width=0.9\columnwidth]{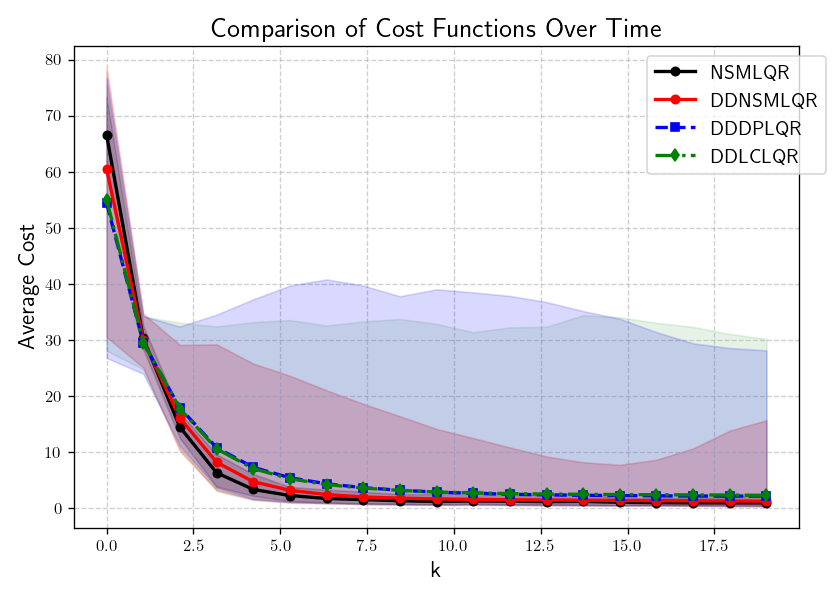}}
    \caption{Average closed-loop cost for all methods on Example 2 with DDDPLQR using \(W = 10^{-3}I\) and DDLCLQR using \(\alpha = 0.01\). }
\label{fig:quan4}
\end{figure}

\begin{table*}[h]
    \centering
    \caption{Performance Comparison of Data-Driven LQR Methods Using Noisy State Measurements for Example 2}
    \label{tab:inverted}
    \renewcommand{\arraystretch}{1.3}
    \resizebox{\textwidth}{!}{%
    \begin{tabular}{clcccc}
        \toprule
        \midrule
        & \textbf{Method} & \textbf{Success Count} & \textbf{Stable Eigenvalues} & $\mathbf{\|K - K^*\|}$ & \textbf{Average Gain Matrix $K$} \\
        \midrule
        \midrule       
        & NSMLQR  & ---  & ---  & --- & \multicolumn{1}{l}{$\begin{bmatrix}  0.13 & -12.34 &     0.44 &  -0.91 \end{bmatrix}$} \rule{0pt}{15pt} \\
        & DDNSMLQR (Theorem \ref{th:ddlqrnoisymeas}) & 150  & 121  & 0.78  & \multicolumn{1}{l}{$\begin{bmatrix}   0.12 & -11.55 &   0.41 &  -0.85 \end{bmatrix}$} \rule{0pt}{15pt} \\
        & DDDPLQR~\cite{ESMZAD2025112197}   & 150  & 2  & 2.04  & \multicolumn{1}{l}{$\begin{bmatrix} 0.14 &  -10.31 &   0.39 &  -0.79 \end{bmatrix}$} \rule{0pt}{15pt} \\
        & DDLCLQR~\cite{DEPERSIS2021109548} $(\alpha=0)$   & 150  & 2  & 3.51  & \multicolumn{1}{l}{$\begin{bmatrix} 0.11 & -8.84 &  0.36 &  -0.67 \end{bmatrix}$} \rule{0pt}{15pt} \\
        \midrule
        & DDDPLQR~\cite{ESMZAD2025112197} $(W=10^{-4}I)$   & 150  & 35  & 1.35  & \multicolumn{1}{l}{$\begin{bmatrix} 0.10 & -10.99 &   0.38 &  -0.79 \end{bmatrix}$} \rule{0pt}{15pt}\\
        & DDLCLQR~\cite{DEPERSIS2021109548} $(\alpha=0.1)$   & 150  & 150  & 0.42  & \multicolumn{1}{l}{$\begin{bmatrix} 0.21 & -12.74 &   0.48 &  -0.95 \end{bmatrix}$} \rule{0pt}{15pt}\\

        \midrule
        & DDDPLQR~\cite{ESMZAD2025112197} $(W=10^{-3}I)$   & 150  & 77  & 1.32  & \multicolumn{1}{l}{$\begin{bmatrix} 0.08 & -11.02 &   0.38 &  -0.80 \end{bmatrix}$} \rule{0pt}{15pt}\\
        & DDLCLQR~\cite{DEPERSIS2021109548} $(\alpha=0.01)$   & 150  & 66  & 1.37  & \multicolumn{1}{l}{$\begin{bmatrix} 0.08 & -10.96 &   0.37 &  -0.79 \end{bmatrix}$} \rule{0pt}{15pt}\\
        
        \midrule
       
        & DDDPLQR~\cite{ESMZAD2025112197}  $(W=10^{-5}I)$  & 150   & 5   & 1.78  & \multicolumn{1}{l}{$\begin{bmatrix} 0.09 & -10.58 &   0.36 &   -0.77 \end{bmatrix}$} \rule{0pt}{15pt} \\
        & DDLCLQR~\cite{DEPERSIS2021109548} $(\alpha=10)$  & 150  & 150  & 2.26   & \multicolumn{1}{l}{$\begin{bmatrix}  0.48 &  -14.54 &   0.63 &  -1.14   \end{bmatrix}$} \rule{0pt}{15pt}\\
        \midrule
        \bottomrule
    \end{tabular}%
    }
\end{table*}




\begin{remark}
    Throughout this work, it is assumed that the noise covariance matrices \( V \) (measurement noise) and \( W \) (process noise) are known. This assumption is common in the data-driven control literature and is also adopted in related works. While knowledge of these covariances facilitates tractable controller synthesis and stability analysis, estimating \( V \) and \( W \) directly from noisy data remains a challenging task. Developing reliable and data-driven methods for covariance estimation constitutes an important direction for future research.
\end{remark}

\section{Conclusion}\label{sec:con}
This paper proposed a direct data-driven control framework for linear systems with noisy state measurements, eliminating the need for explicit system identification. By leveraging noisy input-output trajectories, we developed a convex optimization-based methodology that ensures mean-square stability (MSS) and optimal performance. We first established a theoretical guarantee for MSS preservation in data-driven systems and then introduced a robust LMI-based condition to synthesize stabilizing controllers directly from noisy data. Building upon this, we formulated a novel data-driven LQR problem and solved it via an SDP, yielding optimal feedback gains that closely approximate the model-based solution.

Simulation results on benchmark systems, including an active suspension model and a rotary inverted pendulum, confirmed the robustness and superior performance of our approach compared to existing data-driven LQR methods.

\textbf{Future Work:} While this work assumes knowledge of the noise covariance matrices, estimating these quantities from data remains an open challenge. Future research will focus on developing data-driven noise covariance estimation techniques and integrating them into the control design. Moreover, extending this framework to nonlinear systems, output-feedback scenarios, and systems affected by biased or correlated noise represents promising directions for expanding the applicability of data-driven control in real-world environments.

\bibliographystyle{IEEEtran}
\bibliography{refs} 

\begin{IEEEbiography}[{\includegraphics[width=1in,height=1.25in,clip,keepaspectratio]{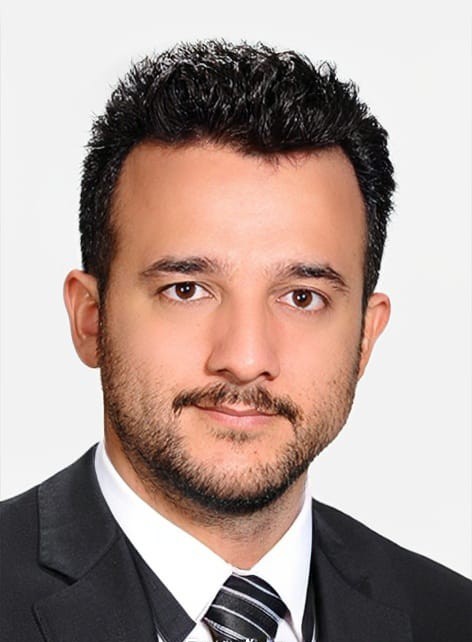}}]{Ramin Esmzad}
received his B.S. and M.S. degrees with honors in Electrical Engineering from Sahand University of Technology, Tabriz, Iran, in 2010 and 2012, respectively. He is currently pursuing his Ph.D. in the Department of Mechanical Engineering at Michigan State University, East Lansing, MI, USA. His research interests include model predictive control, machine learning in control, robotics, data-driven control, reinforcement learning, and cloud-based control system design.
\end{IEEEbiography}

\begin{IEEEbiography}[{\includegraphics[width=1in,height=1.25in,clip,keepaspectratio]{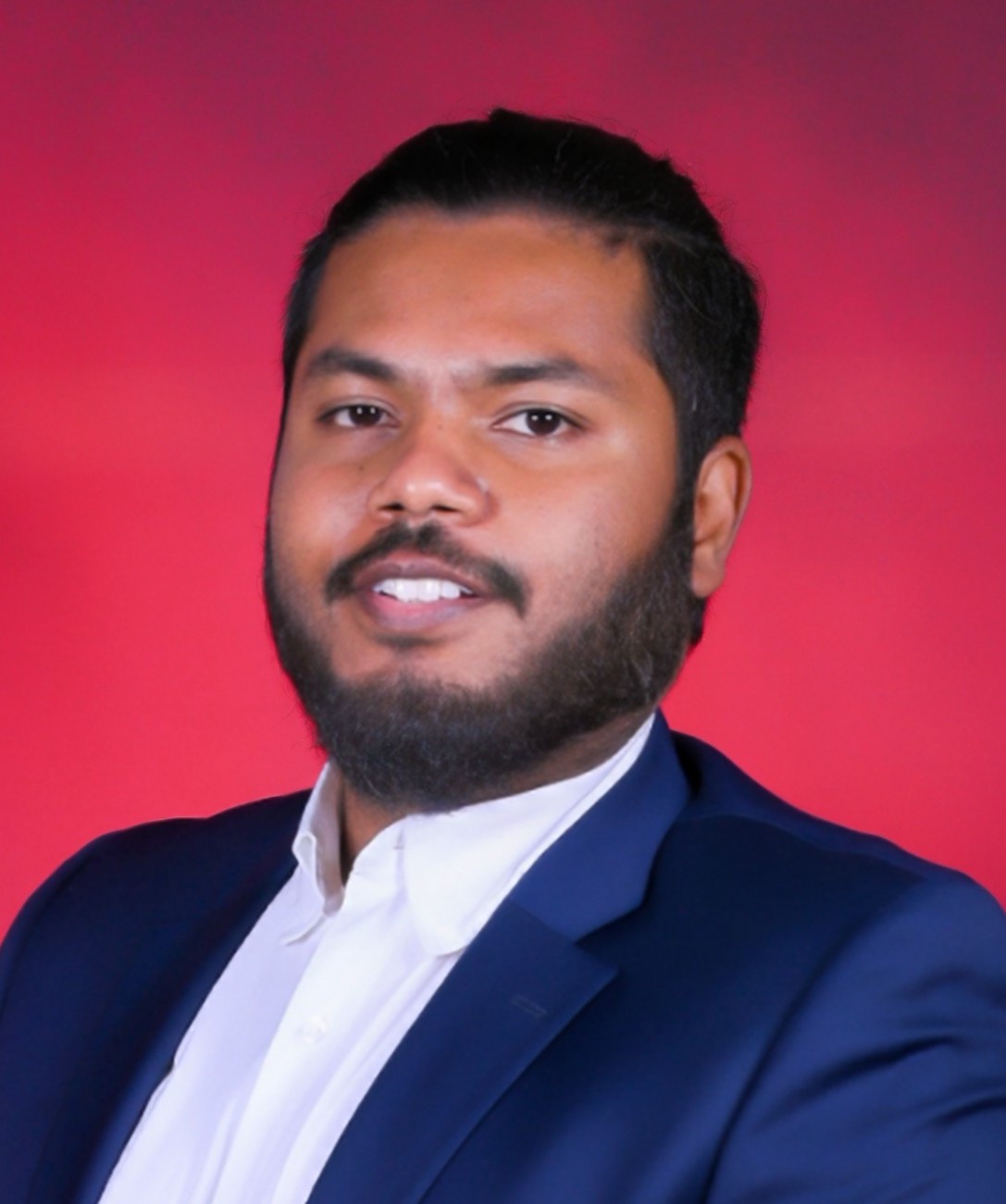}}]{Gokul S. Sankar}
 is a Motion Planning Engineer at Ford Motor Company, where he works at the intersection of cutting-edge autonomy and intelligent control. He earned his Bachelor's degree in Electronics and Instrumentation Engineering from Anna University, Chennai, India, in 2010, followed by an M.S. from the Indian Institute of Technology Madras in 2013, and a Ph.D. from the University of Melbourne, Australia, in 2019.

From 2020 to 2022, Gokul was a Research Engineer at Traxen Inc., USA, where he played a key role in developing advanced model-based control architectures for autonomous long-haul semi-trucks. Prior to that, he was a Research Fellow at the University of Michigan, Ann Arbor, from 2018 to 2019. His interests span reinforcement learning, optimal control, and robust control, with a strong focus on autonomous vehicles and cyber-physical systems.
\end{IEEEbiography}

\begin{IEEEbiography}[{\includegraphics[width=1in,height=1.25in,clip,keepaspectratio]{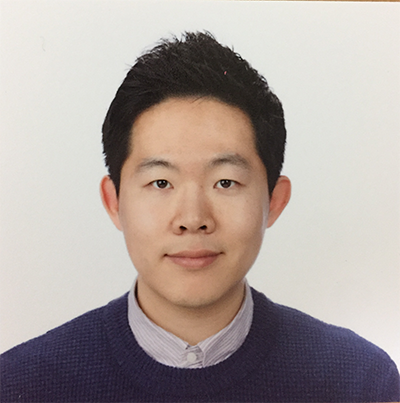}}]{Teawon Han}
 is a Senior Motion Planning Software Engineer at Motional AD Inc., where he develops scalable machine learning algorithms for behavior and trajectory planning in autonomous driving systems. He previously held senior engineering and research roles at Ford Motor Company, contributing to the development of next-generation ADAS technologies and Level 3/4 vehicle automation. Dr. Han earned his Ph.D. in Electrical and Computer Engineering from The Ohio State University, focusing on intelligent control systems for automated vehicles. His research interests include reinforcement learning, evolving systems, motion planning, and interpretable AI.
\end{IEEEbiography}

\begin{IEEEbiography}[{\includegraphics[width=1in,height=1.25in,clip,keepaspectratio]{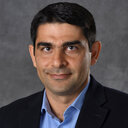}}]{Hamidreza Modares}
 received a B.S. degree from the University of Tehran, Tehran, Iran, in 2004, an M.S. degree from the Shahrood University of Technology, Shahrood, Iran, in 2006, and a Ph.D. degree from the University of Texas at Arlington, Arlington, TX, USA, in 2015, all in Electrical Engineering. He is currently an Associate Professor in the Department of Mechanical Engineering at Michigan State University. Before joining Michigan State University, he was an Assistant professor in the Department of Electrical Engineering at Missouri University of Science and Technology. His current research interests include reinforcement learning, safe control, machine learning in control, distributed control of multi-agent systems, and robotics. He is an Associate Editor of IEEE Transactions on Systems, Man, and Cybernetics: systems. 
\end{IEEEbiography}
\end{document}